\documentclass{llncs}

\usepackage{times}
\usepackage{graphicx}
\usepackage{amsmath}
\usepackage{amssymb}
\usepackage{thmtools}
\usepackage{xspace}
\usepackage[noend]{algpseudocode}
\usepackage{algorithm}
\usepackage{epstopdf}
\usepackage{color}

\algblockdefx[WithProb]{WithProb}{EndProb}[1]{\textbf{with probability} #1 \textbf{do}}{}
\algcblockdefx[WithProb]{WithProb}{OtherwiseProb}{EndProb}{\textbf{otherwise do}}{}
\algtext*{EndProb}

\newcommand{\st}{\;|\;}

\newcommand{\N}{\ensuremath{\mathbb{N}}\xspace}
\newcommand{\NP}{\ensuremath{\mathsf{NP}}\xspace}
\newcommand{\SAT}{\ensuremath{\mathsf{SAT}}\xspace}

\DeclareMathOperator*{\E}{\mathbb{E}}

\newcommand{\ccomm}{C_\text{comm}}
\newcommand{\cbridge}{C_\text{bridge}}
\newcommand{\cadded}{C_\text{added}}

\allowdisplaybreaks

\begin{document}

\title{On the Hardness of SAT with Community Structure}

\author{Nathan Mull \and Daniel J. Fremont \and Sanjit A. Seshia}
\institute{University of California, Berkeley}

\maketitle

\begin{abstract}
Recent attempts to explain the effectiveness of Boolean satisfiability (\SAT) solvers based on conflict-driven clause learning (CDCL) on large industrial benchmarks have focused on the concept of community structure.
Specifically, industrial benchmarks have been empirically found to have good community structure, and experiments seem to show a correlation between such structure and the efficiency of CDCL.
However, in this paper we establish hardness results suggesting that community structure is not sufficient to explain the success of CDCL in practice.
First, we formally characterize a property shared by a wide class of metrics capturing community structure, including ``modularity''.
Next, we show that the {\SAT} instances with good community structure according to any metric with this property are still NP-hard.
Finally, we study a class of random instances generated from the``pseudo-industrial" \emph{community attachment} model of Gir\'aldez-Cru and Levy.
We prove that, with high probability, instances from this model that have relatively few communities but are still highly modular require exponentially long resolution proofs and so are hard for CDCL.
We also present experimental evidence that our result continues to hold for instances with many more communities.
This indicates that actual industrial instances easily solved by CDCL may have some other relevant structure not captured by the community attachment model.
\end{abstract}

\section{Introduction}

Over the last 20 years Boolean satisfiability (\SAT) solvers have become widely used tools for solving problems in many domains~\cite{ms-applications,eda-applications}.
This is largely the result of the \emph{conflict-driven clause learning} (CDCL) paradigm, introduced in the mid-1990s \cite{grasp,bayardo-schrag,sh-cdcl} and much developed since then.
This success of \SAT solving in practice is perhaps surprising in light of the 
{\NP}-hardness of {\SAT}, which is widely interpreted to mean that the problem admits no efficient algorithms.
This has led to a line of research trying to answer the basic question: \emph{why does CDCL perform so well in practice?}
In other words, what is it about industrial \SAT instances that allows them to seemingly avoid the worst-case behavior of CDCL?

One possible explanation is that {\SAT} is significantly easier \emph{on average} than in the worst case.
For algorithms like CDCL that are based on resolution (which we will discuss in more detail below), 
this was ruled out by the discovery that random instances require exponentially-long resolution proofs \cite{exponential-resolution}.
Of course, industrial instances are generally highly \emph{non-random}, so another possibility is that such instances tend to fall into a tractable class of problems.
For example, {\SAT} is known to be fixed-parameter tractable with respect to various natural parameters such as treewidth and clique-width \cite{sh-fpt}.
However, it is unclear whether these parameters are always small in practice.
Moreover, if the goal is to analyze the success of CDCL, the existence of \emph{different} algorithms that take advantage of small (say) treewidth 
is not relevant: what matters is whether it correlates with CDCL performance, and in fact there is evidence against this \cite{treewidth-industrial}.

Parameters more relevant to CDCL are the sizes of \emph{backdoors} \cite{backdoors} and \emph{backbones} \cite{backbones}.
In essence, a backdoor is a set of variables which if assigned cause the instance to become solvable by simplification with no further search, while the backbone is the set of variables which can only be assigned one way in every satisfying assignment.
Correlations between the sizes of backdoors and backbones and the performance of CDCL have been observed empirically, and some ``structured'' instances do seem to have small backdoors \cite{bb-and-bd,empirical-bd}.
Unfortunately, these experiments have all been limited by the computational difficulty of estimating backdoor sizes, 
and it is unclear whether they are representative of a majority of large industrial benchmarks.

None of these ideas have adequately covered the whole variety of industrial instances, leaving a significant gap between our theoretical understanding of when {\SAT} is easy and the reality of CDCL's effectiveness in practice.
Of course, despite the intuition that industrial instances have some common underlying structure that explains why CDCL is so effective on them, it is probable that no single explanation suffices.
There are particular types of industrial instances that are easy for a specific, known reason that does not apply to all other industrial instances \cite{feature-models}.
However, it is still worthwhile to seek general explanations covering as many different types of instances as possible.

One recent approach has focused on the concept of \emph{community structure}, as measured by \emph{modularity} \cite{modularity}.
The variables of an instance with ``good community structure'' (high modularity) can be partitioned into relatively small sets such that few clauses span multiple sets.
It has been found that industrial instances exhibit significantly better community structure than random instances \cite{community-sat}, and that community structure does empirically correlate with CDCL performance \cite{newsham,giraldez2015modularity}.
This makes community structure a plausible candidate for the ``hidden structure'' underlying the effectiveness of CDCL on industrial benchmarks.
In fact, community structure has been used as the basis for a model of random ``pseudo-industrial'' instances, the \emph{community attachment} model \cite{giraldez2015modularity}.
It has parameters controlling the number of communities and the density of their interconnections, allowing it to better reflect the properties of industrial benchmarks.

However, as yet there has been little theoretical analysis connecting community structure to CDCL performance.
The only relevant work we are aware of is that of Ganian and Szeider~\cite{ganian-szeider}, who observe that SAT remains \NP-hard for highly modular instances.
They also give a tractability result for a parameter ``h-modularity" inspired by community structure.
However, this parameter is significantly different from the usual modularity, there is no evidence that it is small in industrial benchmarks, and the tractability result is via an algorithm completely different from CDCL.

In this paper, we extend the connection between community structure and worst-case complexity, and establish the first theoretical result on the average-case performance of CDCL on modular instances.
Specifically, we:
\begin{itemize}
 \item Define the \emph{polynomial clique metrics} (PCMs), a broad class of graph metrics that includes modularity and other popular measures of graph clustering (Section \ref{sec:pcm-definition}).
 \item Show that the set of {\SAT} instances which have ``good community structure'' according to any PCM is still {\NP}-hard (Section \ref{sec:pcm-hardness}).
 \item Prove that on random unsatisfiable instances from the community attachment model that have fewer than $\Theta(n^{1/10})$ communities but can still be highly modular, CDCL takes exponential time with high probability (Section \ref{sec:average-case}).
 \item Give experimental evidence that our result continues to hold for instances with $\Theta(n^\alpha)$ communities for $\alpha \gg 1/10$ (Section \ref{sec:discussion}).
\end{itemize}
Based on these results, we suggest that community structure by itself may not be an adequate explanation for the effectiveness of CDCL in practice.
We begin in Section \ref{sec:background} with background on {\SAT}, CDCL, and community structure both generally and as recently applied to {\SAT}, and conclude in Section \ref{sec:discussion} with a discussion of our results and some directions for future work.

\section{Background} \label{sec:background}

\subsection{SAT}

The \emph{Boolean satisfiability} or {\SAT} problem is to decide, given a Boolean formula $\varphi(\vec{x})$ over a vector of variables $\vec{x}$, whether or not there is a \emph{satisfying assignment} to $\vec{x}$ that makes the formula true.
In this paper, we make the common assumption that the formula $\varphi$ is in \emph{conjunctive normal form} (CNF): it is a conjunction $\psi_1 \land \dots \land \psi_m$ of \emph{clauses}, where each clause $\psi_i$ is a disjunction $\ell_{i1} \lor \dots \lor \ell_{ik}$.
Here each $\ell_{ij}$ is a \emph{literal}: either a variable from $\vec{x}$ or the negation of such a variable.
We also assume that every clause has the same length $k$.
A formula $\varphi$ satisfying these conditions is called a \emph{$k$-\SAT} formula.

Given a partial assignment $\rho$ to some of the variables $\vec{x}$, the \emph{restriction} $\varphi \lceil_\rho$ of $\varphi(\vec{x})$ to $\rho$ is the formula obtained from $\varphi$ by removing all clauses satisfied by $\rho$ and all literals falsified by $\rho$.
We can apply a restriction to any list of clauses analogously.
The \emph{size} of the restriction is the number of variables assigned by $\rho$.

\subsection{Resolution and CDCL}

\emph{Resolution} \cite{resolution} is a fundamental proof system that underlies modern {\SAT} solving algorithms.
It consists of a single rule stating that from clauses $(v \lor \vec{w})$ and $(\lnot v \lor \vec{u})$ that have occurrences of $v$ with opposite polarities, we may infer the clause $(\vec{w} \lor \vec{u})$.
As we will see in a moment, the importance of resolution for our purposes is that in order to establish that a formula $\varphi$ is unsatisfiable, {\SAT} solvers based on CDCL implicitly construct a \emph{resolution refutation} of $\varphi$: a derivation of a contradiction (the empty clause) from $\varphi$ using the resolution rule.
This effectively means that the runtime of such a solver cannot be shorter than the length of the shortest such refutation.

To make this precise we need to define what we mean by CDCL.
\emph{Conflict-driven clause learning} \cite{sh-cdcl} describes a class of algorithms that extend the Davis--Putnam--Logemann--Loveland (DPLL) algorithm \cite{dpll}.
DPLL is a classical search algorithm that assigns each variable in turn, backtracking if a clause is falsified by the current assignment.
If at any point there is a clause with only a single unassigned variable, then that variable can immediately be given the assignment which satisfies the clause --- a rule called \emph{unit propagation}.
If we eventually assign every variable, then we have found a satisfying assignment; otherwise, the search will backtrack all the way to the top level, every possible assignment will have been tried, and the formula is unsatisfiable.

CDCL-type algorithms augment this procedure by \emph{learning} at every backtrack point a new clause that summarizes the reason why the current partial assignment falsifies the formula \cite{grasp,bayardo-schrag}.
This \emph{conflict clause} $C$ is derived by resolving the falsified clause $F$ with one or more other clauses that were used to assign variables in $F$ by unit propagation.
As a result $C$ is always derivable from the original formula $\varphi$ by resolution, and writing out all clauses learned by CDCL when $\varphi$ is unsatisfiable gives a resolution refutation of $\varphi$ \cite{cdcl-sim}.
So the shortest such refutation gives a lower bound for the runtime of CDCL.
This is true regardless of the heuristics used by the particular CDCL variant to decide which variable to assign and its polarity, how exactly to derive the conflict clause, and when to restart search from the beginning (see \cite{cdcl-sim} for a more precise statement).
We also note that pre- or inprocessing techniques that add no clauses (e.g. blocked clause elimination \cite{bce}) or add only clauses derived via resolution (e.g. variable elimination \cite{var-elim}) will not affect the lower bound.

\subsection{Random SAT Instances}

To study the performance of CDCL on ``typical'' instances, we use the framework of \emph{average-case complexity}, which analyzes the efficiency of algorithms on \emph{random} instances drawn from a particular distribution.
We will be interested in complexity lower bounds that hold for almost all sufficiently large instances:
\begin{definition}
An event $X$ occurs \textbf{with high probability} in terms of $n$ if $\Pr[X] \rightarrow 1$ as $n \rightarrow \infty$.
\end{definition}
For example, if flipping $n$ fair coins, with high probability at least $49\%$ will be heads.

Perhaps the simplest distribution over {\SAT} instances arises from fixing the numbers of variables, clauses, and variables per clause, and then sampling uniformly:
\begin{definition}
$F_k(n, m)$ is the uniform distribution over $k$-CNF formulas with $n$ variables and $m$ clauses.
\end{definition}
This \emph{random $k$-{\SAT} model} has been widely studied, and is known to be difficult on average for CDCL (for clause-variable ratios in a certain range) by the resolution lower bound discussed above: with high probability, a random unsatisfiable instance has only exponentially long resolution refutations \cite{beame}.
As we will discuss shortly, our work extends this result to a more recent random {\SAT} model that favors instances that are ``pseudo-industrial'' in the sense of having good community structure.

\subsection{Community Structure}

The notion of community structure has a long history in many fields \cite{community-history}.
The essential idea is that graphs with ``good community structure'' can be broken into relatively small pieces, \emph{communities}, that are densely connected internally but only sparsely connected to each other.
There are a number of metrics which have been proposed to make this notion formal, of which one of the most popular is \emph{modularity} \cite{modularity}.
We consider unweighted graphs as weighted graphs with all weights $1$.
\begin{definition}
Let $G = (V, E)$ and let $\delta = \{C_1, \dots, C_n\}$ be a vertex partition.
Let $\deg v$ be the degree of $v$, and $w(x,y)$ be the weight of the edge $(x,y)$ or zero if there is no such edge.
The \textbf{modularity} (or $Q$-value) of $G$ is
\[
Q = \max_\delta \sum_{C \in \delta} \left[ \frac{\sum_{x, y \in C} w(x, y)}{\sum_{x, y \in V} w(x, y)} - \left(\frac{\sum_{x \in C} \deg x}{\sum_{x \in V} \deg x} \right)^2 \right].
\]
\end{definition}

While work on community structure in {\SAT} instances has focused on modularity, there are several competing metrics that have been used to measure community structure in other domains.
In Appendix \ref{sec:other-metrics} of this paper, we consider four: silhouette index, conductance, coverage, and performance \cite{best-mod}.

Finally, we introduce notation for two graphs that will be useful in this paper: $K_n$, the complete graph on $n$ vertices, and $K_n^m$, consisting of $m$ disjoint copies of $K_n$. 

\subsection{SAT and Community Structure}

Recent work on the community structure of {\SAT} instances begins by associating to each instance its \emph{variable incidence graph} (also known as the \emph{primal graph}).
\begin{definition}
Let $\varphi$ be a CNF formula.
The \textbf{variable incidence graph} (\textbf{VIG}) of $\varphi$ is the graph $G_\varphi = (V, E)$ where $V$ is the set of all variables occurring in $\varphi$ and $E$ is the set $\{(v_1, v_2) : v_1, v_2 \in V$ and they appear together in some clause of $\varphi\}$.
\end{definition}
Some works use a \emph{weighted} version of this graph with $w(v_1,v_2) = \sum_{cl} \left[ 1/\binom{|cl|}{2} \right]$, where the sum is over all clauses in which both $v_1$ and $v_2$ appear \cite{community-sat,giraldez2015modularity}.
This ensures that each clause contributes an equal amount to the total weight of the graph regardless of its length.
Our results apply to both the weighted and unweighted versions.

Obviously, the graph $G_\varphi$ does not preserve all information about the instance $\varphi$.
In particular, the polarities of the literals are ignored.
But the graph does capture significant structural information: for example, if the graph has two connected components on variables $\vec{x}$ and $\vec{y}$ then the formula $\varphi(\vec{x},\vec{y})$ can be split into $\psi(\vec{x}) \land \chi(\vec{y})$ and each subformula solved independently.
In practice a perfect decomposition is rare, but one into \emph{almost} independent parts is more plausible.
This is exactly the idea of community structure, and leads us naturally to consider applying modularity to {\SAT} instances.
\begin{definition}
The \textbf{modularity} of a formula $\varphi$ is the modularity of $G_\varphi$.
\end{definition}

As was mentioned earlier, it has been found empirically that modularity correlates with CDCL performance \cite{newsham}.
This is a claim about the \emph{average} behavior of CDCL over a wide variety of industrial benchmarks, not about its behavior on any specific instance.
Thus it is naturally formalized in the average-case complexity framework discussed above, by giving a distribution that favors instances that are ``industrial'' in character.
One such proposal, based on the idea that the key commonality of industrial instances is their good community structure, is the \emph{community attachment} model of Gir\'aldez-Cru and Levy \cite{giraldez2015modularity}.
In addition to the numbers of variables and clauses, this model has parameters controlling the number of communities and the (expected) fraction of clauses that lie within a single community instead of spanning multiple communities.
\begin{definition}
Let $N$ be a set of $n$ variables. A \textbf{partition of
$N$ into $c$ communities} is a partition $S = \{S_1, \dots, S_c\}$ of $N$ such
that $|S_i| = n / c$. A clause
is \textbf{within a community} if it contains only variables from a single
$S_i$.
A \textbf{bridge clause} is a clause whose variables are all in different communities.
\end{definition}
\begin{definition}[\cite{giraldez2015modularity}] \textbf{(Community Attachment Model)}
Let $n, m, c, k \in \N$ and $p \in [0, 1]$ such that $c$ divides $n$ and $2 \le k \leq c \leq n /
k$.  Then $F_k(n, m, c, p)$ is the distribution over $k$-CNF formulas with $n$
variables and $m$ clauses given by the following procedure: first, choose a random
partition of $n$ variables into $c$ communities. With probability $p$, choose a
clause uniformly among clauses within a community, and otherwise choose uniformly among bridge clauses.
Generate $m$ clauses independently in this way.
\end{definition}
\begin{remark}
We define bridge clauses in a way that matches the community attachment model, but as we will discuss below our results also hold for a modified model where a bridge clause is any clause not within a single community.
\end{remark}

Like the random $k$-{\SAT} model $F_k(n,m)$, the model $F_k(n,m,c,p)$ ranges over $k$-CNF formulas with $n$ variables and $m$ clauses, and each clause is chosen independently of the others.
However, in this model the clauses are of two different types: those lying entirely within a community, and those spread across $k$ different communities.
The probability $p$ controls how likely a clause is to be of the first type versus the second.

The idea behind this model is that by picking $c$ and $p$ appropriately, one is likely to obtain instances that decompose into loosely-connected communities, as has been observed in actual industrial instances.
More precisely, the expected modularity of an instance drawn from $F_k(n,m,c,p)$ is lower bounded by $p - (1/c)$, so that for nontrivial $c$ highly modular instances can be generated by setting $p$ large enough \cite{giraldez2015modularity}.
Furthermore, Gir\'aldez-Cru and Levy find experimentally that high-modularity instances generated with this model are solved more quickly by CDCL than by look-ahead solvers, and the reverse is true for low-modularity instances \cite{giraldez2015modularity}.
This parallels the same observation for industrial instances versus random instances.
Thus, they conclude, $F_k(n,m,c,p)$ is a more realistic model of industrial instances than the random $k$-{\SAT} model $F_k(n,m)$.

\section{Worst-Case Hardness} \label{sec:worst-case}

In this section, we propose a simple class of graph metrics that we argue should include most metrics quantifying community structure.
We show that modularity is in fact within the class, as are several other popular graph clustering metrics.
However, we demonstrate that the set of $\SAT$ instances that have ``good community structure'' according to any metric in the class is \NP-hard.
Therefore, no such metric can be a guaranteed indicator of the difficulty of a $\SAT$ instance.

\subsection{A Class of ``Modularity-like'' Graph Metrics} \label{sec:pcm-definition}

We begin by formalizing what we mean by a graph metric.
\begin{definition}
A \textbf{graph metric} is a function $m$ from weighted graphs to $[0,1]$.
Given $m$ and any $\epsilon \in [0,1]$,  $\SAT_{m,\epsilon}$ is the class of all \SAT instances $\varphi$ such that $m(G_\varphi) \geq 1 - \epsilon$.
\end{definition}
For example, if $m$ is modularity then $\SAT_{m,\epsilon}$ consists of the ``high modularity'' formulas, where ``high'' means any modularity above $1-\epsilon$.

In general we are interested in graph metrics that represent a notion of community structure, assigning larger values to graphs which have such a structure than those that do not.
For such a metric $m$, consider the following property:
\begin{definition}
A graph metric $m$ is a \textbf{polynomial clique metric} (\textbf{PCM}) if for all $\epsilon > 0$, there is a poly-time computable function $c : \N \rightarrow \N$ with at most polynomial growth and some $n_0 \in \N$ such that for all $n \ge n_0$, if $K$ is $K_n$ with any positive edge weights then $m(K^{c(n)}) \ge 1 - \epsilon$.
\end{definition}
\begin{remark}
If using the unweighted version of the variable incidence graph, our proofs will work using a relaxed definition that applies only to $K = K_n$ with unit weights.
\end{remark}
In essence, the definition states that for any (sufficiently large) size $n$, at most a polynomial number of copies of $K_n$ are needed to produce a graph that $m$ considers to have ``good community structure''.
This is a natural property for modularity-like metrics to have, since copies of $K_n$ are in some sense ideal communities: internally connected as much as possible, with no external edges.
Of course we would not consider a single copy of $K_n$ to have good community structure, so the definition of a PCM only requires that such structure be obtained for \emph{some} number of copies at most polynomial in $n$.

Next we demonstrate that the PCMs are a large class including modularity and several other popular clustering metrics.
While the other metrics have not been experimentally evaluated in the context of SAT, this still supports our claim that the PCM property is a natural one for metrics of community structure to have.
For lack of space, we defer the definitions and analysis of the metrics other than modularity to Appendix \ref{sec:other-metrics}.
\begin{theorem}
Modularity is a PCM.
\end{theorem}
\begin{proof}
Fix any $\epsilon > 0$ and $n \geq 2$.
Let $K$ be $K_n$ with arbitrary positive edge weights, and let $G = K^c$.
Let $\delta$ be the vertex partition that groups two vertices iff they are in the same copy of $K$.
Then since each community is identical, and there are $c$ communities,
$\sum_{x, y \in C} w(x, y) / \sum_{x, y \in V} w(x, y) = 1/c$ and $\sum_{x \in C} \deg x / \sum_{x \in V} \deg x = 1/c$
for any $C \in \delta$.
Therefore, $Q(G) \geq c(1 / c - (1 / c)^2) = 1 - 1 / c$.
Putting $c = 1 / \epsilon$, we have $Q(G) \geq 1 - \epsilon$.
Since $c$ is $O(1)$ with respect to $n$, $Q$ is a PCM.
\qed
\end{proof}

\vspace{-1.7ex} 
\begin{restatable}{theorem}{otherMetricsPCM}
Silhouette index, conductance, coverage, and performance are PCMs.
\end{restatable}

\subsection{Hardness of PCM-Modular Instances} \label{sec:pcm-hardness}

Now we show that the $\SAT$ instances which have ``good community structure'' according to a PCM are no easier in the worst case than any other instance.
The PCMs thus form a wide class of metrics which cannot be used as a guaranteed indicator of the difficulty of a $\SAT$ instance.
Our reduction can be viewed as a variation of that suggested by Ganian and Szeider~\cite{ganian-szeider} to show \NP-hardness in the specific case of modularity.
\begin{theorem}
For any PCM $m$, the class $\SAT_{m,\epsilon}$ is \NP-hard for all $\epsilon > 0$.
\end{theorem}
\begin{proof}
Given a $\SAT$ instance $\phi$, we will convert it into an
equisatisfiable instance of $\SAT_{m,\epsilon}$ in polynomial time.  Let $V$ be
the set of all variables occurring in $\phi$, along with new variables as necessary so that $|V| \ge n_0$.
Fixing a variable $x$ not in $V$, let $\psi$ be the formula obtained by adding to $\phi$ all clauses of the
form $x \lor y \lor z$ with $y, z \in V$.  Clearly, the VIG of $\psi$ is $K_n$ with $n = |V| + 1 \ge n_0$.
Furthermore, $\phi$ and $\psi$ are equisatisfiable, since we can simply assert $x$ to
satisfy all the new clauses.  Now letting $\chi$ be the conjunction of $c(n)$
disjoint copies of $\psi$ (i.e. copies with variables renamed so none are
common), the variable incidence graph $G$ of $\chi$ is $K_n^{c(n)}$ (with some positive weights).
By the PCM property, we have $m(G) \ge 1-\epsilon$, so $\chi \in
\SAT_{m,\epsilon}$.  Since $\chi$ and $\psi$ are clearly equisatisfiable, so are
$\chi$ and $\phi$, and thus this procedure gives a reduction from $\SAT$ to
$\SAT_{m,\epsilon}$.  Finally, the procedure is polynomial-time since $c(n)$ has
at most polynomial growth and can be computed in polynomial time.
\qed
\end{proof}

\section{Average-Case Hardness} \label{sec:average-case}

In contrast to the previous section, we now consider the difficulty of modular instances for a particular class of algorithms, namely those like CDCL which prove unsatisfiability by effectively constructing a resolution refutation.
While these results are therefore more specific, they are also much more powerful: they show that modular instances are difficult not just in the worst case but also on average.

Our argument is largely based on the resolution lower bound of Beame and Pitassi \cite{beame}, which can be used to establish the hardness of instances from the random $k$-{\SAT} model.
In order to use that result, we need to show that most instances from the community attachment model have certain \emph{sparsity} properties used by the proof.
So our main steps, detailed in Sections \ref{sec:defining-distribution}--\ref{sec:cdcl-runtime} below, are as follows:
\begin{enumerate}
 \item Define a new distribution $\overline{F}_k(n,m,c,p;m')$ over $k$-CNF formulas that works by taking a \emph{random subformula} of an instance from the random $k$-{\SAT} model $F_k(n,m')$.
 \item Show that this new distribution is in fact identical to the community attachment model $F_k(n,m,c,p)$.
 \item Observe that the sparsity properties are inherited by subformulas, so the sparsity result in \cite{beame} for the random $k$-{\SAT} model $F_k(n,m')$ transfers to the community attachment model $F_k(n,m,c,p)$.
 \item Adapt the Beame--Pitassi argument \cite{beame} to obtain an exponential lower bound on the resolution refutation length.
 \item Conclude that CDCL takes exponential time on unsatisfiable formulas from the community attachment model $F_k(n,m,c,p)$ with high probability.
\end{enumerate}

\subsection{Defining the New Distribution} \label{sec:defining-distribution}

We begin by defining our new distribution $\overline{F}_k(n,m,c,p;m')$, which takes an additional parameter $m'$ that we will specify in Section \ref{sec:transfer}.
\begin{definition} \label{def:mod-levy-dist}
Let $n, m, c, k, m' \in \N$ and $p \in [0, 1]$ such
that $2 \le k \leq c \leq n / k$. Then $\overline{F}_k(n, m, c, p;
m')$ is the distribution over $k$-CNF formulas with $n$ variables and $m$
clauses defined by Algorithm \ref{mod_levy} (which is such a distribution by Lemma \ref{lemma:well-defined} below).
\end{definition}

\begin{algorithm}
\caption{defining the distribution $\overline{F}_k(n, m, c, p; m')$}
\label{mod_levy}
\begin{algorithmic}[1]
\State choose $\phi$ from $F_k(n, m')$ \label{line:sample-phi}
\State choose a uniformly random partition of the $n$ variables into $c$ communities
\State $h \gets c\binom{n / c}{k} / \binom{n}{k}$
\State $b \gets (n / c)^k\binom{c}{k} / \binom{n}{k}$
\State $\psi \gets$ the empty formula on $n$ variables
\ForAll {clauses $C$ of $\phi$}
    \WithProb {$p$} \label{line:loop-start}
        \If {$C$ is within a community}
            \State add $C$ to $\psi$
        \EndIf
    \OtherwiseProb \label{line:otherwise-branch}
        \WithProb {$h / b$} \label{line:second-coin}
            \If {$C$ is a bridge clause}
                \State add $C$ to $\psi$
            \EndIf
        \EndProb
    \EndProb
    \If {$|\psi| = m$}
        \Return $\psi$ \Comment{the algorithm ``succeeds''} \label{line:success}
    \EndIf
\EndFor
\State choose a fresh $\psi$ from $F_k(n, m, c, p)$ \label{line:sample-psi} \\
\Return $\psi$ \Comment{the algorithm ``fails''} \label{line:failure}
\end{algorithmic}
\end{algorithm}

\begin{lemma} \label{lemma:well-defined}
For all parameters satisfying the conditions of Definition \ref{def:mod-levy-dist}, Algorithm \ref{mod_levy} defines a probability distribution over $k$-CNF formulas with $n$ variables and $m$ clauses.
\end{lemma}
\begin{proof}
First we must check that $h/b \le 1$ so that the algorithm is well-defined.
We have
\[
\frac{h}{b} = \frac{c \binom{n/c}{k}}{\left( \frac{n}{c} \right)^k \binom{c}{k}} \le \frac{c \left( \frac{n}{c} \right)^k}{k! \left( \frac{n}{c} \right)^k \left( \frac{c}{k} \right)^k} = \frac{k^k}{k! \; c^{k-1}} \le \frac{1}{(k-1)!} \le 1,
\]
since we assume $c \ge k$.
Algorithm \ref{mod_levy} always terminates, returning a formula $\psi$ from either line \ref{line:success} or \ref{line:failure}.
In the first case, $\psi$ is a subset of $\phi$, which is drawn from $F_k(n,m')$ and so has $k$-CNF clauses over $n$ variables.
Furthermore, the algorithm does not return from line \ref{line:success} unless $\psi$ has $m$ clauses.
In the second case, $\psi$ is drawn from $F_k(n,m,c,p)$, and so again is a $k$-CNF formula with $n$ variables and $m$ clauses.
\qed
\end{proof}

\subsection{Comparing the Distribution to the Community Attachment Model}

Next we prove that our definition via Algorithm \ref{mod_levy} is equivalent to the usual community attachment definition.
Since the algorithm adds each clause independently, in essence this amounts to showing that each clause is within a community with probability $p$.
\begin{lemma} \label{lemma:equiv}
For any $m' \in \N$, the distribution $\overline{F}_k(n, m, c, p ; m')$ is identical to the distribution $F_k(n, m, c, p)$.
\end{lemma}
\begin{proof}
When Algorithm \ref{mod_levy} returns a formula $\psi$ from line \ref{line:failure}, $\psi$ is drawn from $F_k(n,m,c,p)$, and so the two distributions are trivially identical.
So we need only consider the case when the algorithm returns from line \ref{line:success}.
Because the algorithm handles each clause of $\phi$ independently (until $m$ clauses are added), it suffices to show that when a clause is added to $\psi$, it is within a community with probability $p$ and is otherwise a bridge clause.
Starting from line \ref{line:loop-start} of Algorithm \ref{mod_levy}, let $\ccomm$ be the event that the clause $C$ is within a community, $\cbridge$ the event that $C$ is a bridge clause, and $\cadded$ the event that $C$ is added to $\psi$.
Let $A$ be the event that the algorithm takes the random branch on line \ref{line:loop-start} instead of the branch on line \ref{line:otherwise-branch}.
Then we have
\[
\Pr[ \ccomm | \cadded ] = \Pr[ \ccomm | A, \cadded] \Pr[A | \cadded] + \Pr[ \ccomm | \overline{A}, \cadded] \Pr[ \overline{A} | \cadded] .
\]
The second term is zero because $\overline{A}$ means the algorithm takes the branch on line \ref{line:otherwise-branch} and thus only adds the clause if it is a bridge clause.
Likewise, $\Pr[ \ccomm | A, \cadded] = 1$ because the branch on line \ref{line:loop-start} only adds the clause if it is within a community.
So
\[
\Pr[ \ccomm | \cadded ] = \Pr[A | \cadded] = \Pr[\cadded | A] \Pr[A] \;/\; \Pr[\cadded] .
\]
By straightforward counting arguments, $\Pr[\ccomm] = c\binom{n / c}{k} / \binom{n}{k} = h$ and $\Pr[\cbridge] = (n / c)^k\binom{c}{k} / \binom{n}{k} = b$.
Since the coin flips on lines \ref{line:loop-start} and \ref{line:second-coin} are independent of $C$, we have $\Pr[\cadded | A] = \Pr[\ccomm] = h$ and $\Pr[\cadded | \overline{A}] = (h/b) \Pr[\cbridge] = h$.
Also $\Pr[A] = p$, so
\[
\Pr[\cadded] = \Pr[\cadded | A] \Pr[A] + \Pr[\cadded | \overline{A}] \Pr[\overline{A}] = h p + h (1-p) = h .
\]
Plugging these into the expression above we obtain $\Pr[ \ccomm | \cadded ] = p$.
So each clause added to $\psi$ is within a community with probability $p$, and otherwise by construction it must be a bridge clause.
Therefore when Algorithm \ref{mod_levy} returns from line \ref{line:success}, it is equivalent to generating $m$ clauses independently, each of which is a uniformly random clause within a community with probability $p$, and otherwise a uniformly random bridge clause.
So $\overline{F}_k(n, m, c, p ; m')$ is identical to $F_k(n,m,c,p)$.
\qed
\end{proof}

\subsection{Transferring Subformula-Inherited Properties} \label{sec:transfer}

Algorithm \ref{mod_levy} can ``fail'' by adding fewer than the desired number of clauses $m$ to $\psi$, then falling back on the community attachment model as a backup.
Otherwise, the algorithm ``succeeds'', returning on line \ref{line:success} a formula that was built up from clauses of $\phi$ and is therefore a subformula of it.
Since our goal is to have the formulas from this distribution inherit properties from $\phi$, we need to ensure that Algorithm \ref{mod_levy} succeeds with high probability.
We can do this by taking $m'$, the number of clauses in $\phi$, to be large enough: then even if a given clause is only added to $\psi$ with a small probability, overall we are likely to add $m$ of them.
As we will see in the proof, the probability of adding a clause is roughly $1/c^{k-1}$, so taking $m'$ to be slightly larger than $c^{k-1} m$ will suffice.
We use the following standard tail bound.
\begin{lemma} \label{bound} If $B(n,p)$ is the number of successes in $n$ Bernoulli trials each with success probability $p$, then for $k < pn$ we have
\[
\Pr[B(n, p) \leq k] \leq \exp\left(\frac{-(pn - k)^2}{2pn}\right).
\]
\end{lemma}

\begin{lemma} \label{lemma:success}
Suppose that $c$ is $o(n)$, $m \rightarrow \infty$ as $n \rightarrow \infty$, and $m' = (1 + \epsilon)c^{k - 1}m$ for some $\epsilon > 0$.
Then Algorithm \ref{mod_levy} returns from line \ref{line:success} with high probability.
\end{lemma}
\begin{proof}
As shown in Lemma \ref{lemma:equiv}, the probability that starting from line \ref{line:loop-start} the clause $C$ will be added to $\psi$ is $h$.
So the probability that Algorithm \ref{mod_levy} returns from line \ref{line:success} is $\Pr[B(m', h) \geq m] = 1 - \Pr[B(m', h) \leq m - 1]$.
Now observe that
\[
h c^{k-1} = \frac{c^k\binom{n /c}{k}}{\binom{n}{k}} = \frac{n (n-c) \cdots (n-c(k+1))}{n (n-1) \cdots (n-k+1)} \le 1.
\]
Furthermore, we have
\[
\lim_{n \rightarrow \infty} h c^{k-1} = \lim_{n \rightarrow \infty} \frac{c^k\binom{n /c}{k}}{\binom{n}{k}} = \lim_{n \rightarrow \infty} \left[ \frac{\binom{n/c}{k}}{\frac{(n/c)^k}{k!}} \cdot \frac{\frac{n^k}{k!}}{\binom{n}{k}} \right] = \left[ \lim_{n \rightarrow \infty} \frac{\binom{n/c}{k}}{\frac{(n/c)^k}{k!}} \right] \left[ \lim_{n \rightarrow \infty} \frac{\frac{n^k}{k!}}{\binom{n}{k}} \right] = 1,
\]
where in evaluating the second-to-last limit we use the fact that $c$ is $o(n)$ and so $\lim_{n \rightarrow \infty} (n/c) = \infty$.
So for sufficiently large $n$ we have $hc^{k - 1} \ge 1 - \epsilon / 2(1 + \epsilon)$, and therefore
\[
h m' = h(1 + \epsilon)c^{k - 1}m \ge \left(1 - \frac{\epsilon}{2(1 + \epsilon)} \right) (1 + \epsilon) m = (1 + \epsilon/2) m.
\]
Applying Lemma \ref{bound}, we have
\begin{align*}
\Pr[B(m',h) \le m - 1] &\le \exp \left( \frac{-[hm' - (m-1)]^2}{2 h m'} \right) \le \exp \left( \frac{-[(1+\epsilon/2) m - m]^2}{2 h (1+\epsilon) c^{k-1} m} \right) \\
&= \exp \left( \frac{-m (\epsilon/2)^2}{2(1+\epsilon) \cdot h c^{k-1}} \right) \le \exp \left( \frac{-m \epsilon^2}{8 (1+\epsilon)} \right),
\end{align*}
which goes to zero as $m \rightarrow \infty$, and therefore as $n \rightarrow \infty$.
So with high probability, Algorithm \ref{mod_levy} will return from line \ref{line:success}.
\qed
\end{proof}

Now it is simple to show that subformula-inherited properties are indeed passed down from random $k$-{\SAT} instances to instances drawn from our distribution.
Here ``subformula-inherited'' simply means that if $\varphi$ has the property, then any formula made up of a subset of the clauses of $\varphi$ also has the property.
For example, being satisfiable is subformula-inherited, but being unsatisfiable is not.
\begin{lemma} \label{lemma:inherit}
Suppose that $c$ is $o(n)$, $m \rightarrow \infty$ as $n \rightarrow \infty$, $m' = (1 + \epsilon)c^{k - 1}m$ for some $\epsilon > 0$, and $P$ is a subformula-inherited property.
Then if a formula drawn from $F_k(n,m')$ has property $P$ with high probability, a formula drawn from $\overline{F}_k(n, m, c, p ; m')$ has property $P$ with high probability.
\end{lemma}
\begin{proof}
Run Algorithm \ref{mod_levy} to sample from $\overline{F}_k(n, m, c, p; m')$.
Let $P_\psi$ and $P_\phi$ respectively be the events that the returned formula $\psi$ and the formula $\phi$ from line \ref{line:sample-phi} have property $P$.
Also let $R$ be the event that the algorithm returns from line \ref{line:success}.
When the algorithm returns from line \ref{line:success}, $\psi$ is a subformula of $\phi$, and since $P$ is inherited by subformulas we have $\Pr[ P_\psi | R ] \ge \Pr[ P_\phi ]$.
Now as $\phi$ is drawn from $F_k(n,m')$, the event $P_\phi$ occurs with high probability, and so $\Pr[ P_\psi | R ] \rightarrow 1$ as $n \rightarrow \infty$.
By Lemma \ref{lemma:success}, the event $R$ also happens with high probability, so $\Pr[ P_\psi ] \ge \Pr[ P_\psi \land R ] = \Pr[ P_\psi | R ] \cdot \Pr[R] \rightarrow 1$ as $n \rightarrow \infty$.
Therefore $\psi$ has property $P$ with high probability.
\qed
\end{proof}

Together, Lemmas \ref{lemma:equiv} and \ref{lemma:inherit} show that subformula-inherited properties of random $k$-{\SAT} instances are also possessed (with high probability) by instances from the community attachment model.

\subsection{Proving the Resolution Lower Bounds} \label{sec:beame-pitassi}

Now we transition to adapting the argument of Beame and Pitassi \cite{beame}.
The proof uses two types of sparsity conditions.
Both view a clause $C$ as a set of variables, so that another set of variables $X$ ``contains'' $C$ if and only if every variable in $C$ is in $X$.
\begin{definition}
A formula is \textbf{$n'$-sparse} if every set of $s \leq n'$ variables contains at most $s$ clauses. \end{definition}
\begin{definition}
Let $n' < n''$. A formula is \textbf{$(n', n'', y)$-sparse} if every set of $s$ variables with $n' < s \leq n''$ contains at most $ys$ clauses.
\end{definition}
These are both clearly subformula-inherited.

The Beame--Pitassi argument \cite{beame} is broken into three major lemmas, each of which we will use without change.
The last lemma establishes the sparsity properties above for the random $k$-{\SAT} model.
\begin{lemma}[\cite{beame}] \label{big-clause} Let $n' \leq n$ and $F$ be an unsatisfiable
CNF formula in $n$ variables with clauses of size at most $k$ that is both
$n'$-sparse and $(n'(k + \epsilon) / 4, n'(k + \epsilon) / 2, 2 / (k +
\epsilon))$-sparse. Then any resolution proof $P$ of the unsatisfiability of $F$
must include a clause of length at least $\epsilon n' / 2$. \end{lemma}

\begin{lemma}[\cite{beame}] \label{whp-proof} Let $P$ be a resolution refutation of $F$ of
size $S$. Given $\beta > 0$, say the \textbf{large clauses} of $P$ are those clauses mentioning more than $\beta n$
distinct variables.  Then with probability at least $1 - 2^{1 - \beta t / 4}S$, a
random restriction of size $t$ sets all large clauses in $P$ to 1. \end{lemma}

\begin{lemma}[\cite{beame}] \label{sparsity} Let $x > 0$, $1 \geq y > 1 / (k - 1)$, and $z
\geq 4$. Fix a restriction $\rho$ on $t \leq \min\{xn / 2, x^{1 - 1 / y(k -
1)}n^{1 - 1/(k - 1)} / z\}$ variables. Drawing $F$ from $F_k(n,m)$ with
\[
m \leq \frac{y}{e^{1 + 1 / y}2^{k + 1 / y}}x^{1/ y - (k - 1)}n ,
\]
then with probability at least $1 - 2^{-t} - (2^k + 1) / z^{k - 1}$, $F\lceil_\rho$ is both $(xn / 2, xn, y)$-sparse and $xn$-sparse.
\end{lemma}

We can now combine these to prove the analog of the main theorem of Beame and Pitassi for modular instances.
Our argument is almost identical to theirs: the only difference is that we apply Lemma \ref{sparsity} to larger instances from the random $k$-{\SAT} model,\footnote{Note that as required by its statement, we are applying Lemma \ref{sparsity} to formulas drawn from $F_k(n,m)$, \emph{not} to formulas drawn from $\overline{F}_k(n,m,c,p;m')$. Lemmas \ref{big-clause} and \ref{whp-proof} work for any formula, so we may use all three lemmas precisely as proved in \cite{beame}.} so that our results above will give us sparsity for modular instances of the correct size embedded in them as subformulas.
\begin{theorem} \label{main}
Let $k \geq 3$, $0 < \epsilon < 1$, and $x$, $t$, $z$, $c$ be functions of $n$ such that $x > 0$, $t$ and $z$ are
$\omega(1)$, $c$ is $o(n)$, and $t$ satisfies the conditions of Lemma \ref{sparsity} for all
sufficiently large $n$. Then with high probability, an unsatisfiable formula drawn from $F_k(n, m, c, p)$ with
\[
m \leq \frac{1}{2^{7k / 2}(1 + \epsilon)c^{k - 1}} x^{-(k - 2 - \epsilon) / 2} n
\]
does not have a resolution refutation of size $\leq 2^{\frac{\epsilon}{4(k + \epsilon)}xt} / 8$.
\end{theorem}
\begin{proof}
Let $S = 2^{\frac{\epsilon}{4(k +
\epsilon)}xt} / 8$ and let $U$ be the set of unsatisfiable $k$-CNF formulas with $n$ variables and $m$ clauses. For each $\varphi \in U$ fix a shortest
resolution refutation $P_\varphi$, and let $W \subseteq U$ be the set of
$\varphi$ such that $|P_\varphi| \leq S$.
Let $R$ be the set of all restrictions of size $t$, and for any formula $\varphi$ and $\rho \in R$ let $L(\varphi, \rho)$ be the indicator function for the event that either $\varphi$ is satisfiable or $P_\varphi \lceil_\rho$ contains a clause of length at least $\epsilon x n / (k + \epsilon)$.
Now for any $\varphi \in W$, by Lemma \ref{whp-proof} with $\beta = \epsilon x / (k + \epsilon)$ we have
\[
\sum_\rho \frac{L(\varphi, \rho)}{|R|} \le 2^{1 - \frac{\epsilon}{4(k + \epsilon)} xt}S = 2^{1 - \frac{\epsilon}{4(k +
\epsilon)} xt}(2^{\frac{\epsilon}{4(k + \epsilon)} xt} / 8) = 1 / 4 .
\]
Let $X$ be a random variable defined over a restriction $\rho$ and equal to $\Pr_\varphi [ L(\varphi, \rho) | \varphi \in W]$, where $\varphi$ is distributed as $F_k(n,m,c,p)$.
Putting a uniform distribution on $\rho$ and writing $q(\psi)$ for the conditional distribution $\Pr_\varphi [ \varphi = \psi | \varphi \in W ]$,
\[
\E_\rho [ X ] = \sum_\rho \frac{1}{|R|} \Pr_\varphi [ L(\varphi, \rho)  | \varphi \in W ] = \sum_{\psi \in W} q(\psi) \left[ \sum_\rho \frac{L(\psi, \rho)}{|R|} \right] \le \sum_{\psi \in W} q(\psi) \frac{1}{4} = \frac{1}{4} .
\]
So by Markov's inequality,
\[
\Pr_\rho [X \geq 1 / 2] \leq \frac{\E_\rho[X]}{1 / 2} \leq 1 / 2 ,
\]
and therefore there is some $\rho'$ such that $\Pr_\varphi [ L(\varphi, \rho') | \varphi \in W ] \le 1/2$.
In other words, there is a restriction that eliminates large clauses from a random $\varphi \in W$ with probability at least $1/2$.

Now let $y = 2 / (k + \epsilon)$.
Since $k \ge 3$ and $\epsilon < 1$ we have $y \ge 1 / (k-1)$ and
\begin{align*}
\frac{y}{e^{1 + 1 / y}2^{k + 1 / y}} &= 2 \left[ (k+\epsilon) e^{1 + \frac{k+\epsilon}{2}} 2^{k + \frac{k+\epsilon}{2}} \right]^{-1} \ge 2 (k+\epsilon)^{-1} e^{-\frac{k}{2} - \frac{3}{2}} 2^{-\frac{3k}{2} - \frac{1}{2}} \\
&= 2 (k+\epsilon)^{-1} e^{-3/2} 2^{-1/2} 2^{- k (3 + \log_2 e) / 2} \\
&\ge 2 (k+1)^{-1} e^{-3/2} 2^{-1/2} 2^{- 2.23 k} \ge 2^{-1.23 k} 2^{- 2.23 k} \ge 2^{-7k/2} .
\end{align*}
By our assumption on $m$,
\begin{align*}
(1+\epsilon) c^{k-1} m &\le 2^{-7k/2} x^{-(k-2-\epsilon)/2} n = 2^{-7k/2} x^{1/y - (k-1)} n \\
&\le \frac{y}{e^{1 + 1 / y}2^{k + 1 / y}} x^{1/y - (k-1)} n .
\end{align*}
Finally, since $z$ is $\omega(1)$ we have $z \geq 4$ for sufficiently large $n$, and then all the conditions of Lemma \ref{sparsity} are satisfied by $y$, $z$, $t$,
and $m' = (1 + \epsilon)c^{k - 1}m$.
Therefore for a formula $\varphi$ drawn from $F_k(n, m')$, $\varphi \lceil_{\rho'}$ is simultaneously $(xn / 2, xn, 2 / (k + \epsilon))$-sparse and $xn$-sparse with probability at least $1 - 2^{-t} - (2^k + 1) / z^{k - 1}$.
Since $t$ and $z$ are $\omega(1)$, $\varphi$ has this property with high probability.
Furthermore, the property is inherited by subformulas, so by Lemma \ref{lemma:inherit} it also holds with high probability for formulas drawn from $\overline{F}_k(n, m, c, p; m')$.
Then by Lemma \ref{lemma:equiv} the same is true for formulas drawn from $F_k(n, m, c, p)$.

Now let $n' = 2xn / (k + \epsilon)$.
Since $k + \epsilon \ge 3$, we have $n' \le xn$ and so $xn$-sparsity implies $n'$-sparsity.
Also note that
\[
\frac{xn}{2} = \frac{2xn(k + \epsilon)}{4(k + \epsilon)} = \frac{n'(k + \epsilon)}{4}
\hspace{0.5cm}\text{and}\hspace{0.5cm}
xn = \frac{n'(k + \epsilon)}{2} .
\]
So by Lemma \ref{big-clause}, when drawing an unsatisfiable formula $\varphi$ from $F_k(n,m,c,p)$, with high probability every resolution refutation of $\varphi \lceil_{\rho'}$ has a clause of length at least $\epsilon n' / 2 = \epsilon xn / (k + \epsilon)$.
That is, $\Pr_\varphi [ L(\varphi, \rho') \st \varphi \in U ] \rightarrow 1$ as $n \rightarrow \infty$.
So
\[
\Pr_\varphi [ \varphi \in W | \varphi \in U] = \frac{\Pr_\varphi [\varphi \in W \land \overline{L(\varphi,\rho')} \;|\; \varphi \in U]}{\Pr_\varphi [\overline{L(\varphi,\rho')} \;|\; \varphi \in W]} \le \frac{\Pr_\varphi [\overline{L(\varphi, \rho')} \;|\; \varphi \in U]}{1/2} \rightarrow 0
\]
as $n \rightarrow \infty$.
Therefore with high probability, an unsatisfiable instance drawn from $F_k(n,m,c,p)$ does not have a resolution refutation of size $\le S$.
\qed
\end{proof}

Next we instantiate this general result to obtain exponential lower bounds for the refutation length when the number of communities is not too large.
We use slightly different arguments for $k \ge 4$ and $k = 3$, again following Beame and Pitassi \cite{beame}.
As the computations are uninteresting, we defer the proofs to Appendix \ref{sec:more-proofs}.
The basic idea is to let $x$ go to zero fast enough that the bound on $m$ required by Theorem \ref{main} is satisfied when $m = O(n)$, but slowly enough that the length bound is of the form $2^{O\left( n^\lambda \right)}$.
\begin{restatable}{theorem}{thmBoundKFour} \label{thm:k4}
Suppose that $k \ge 4$, $m = O(n)$, and $c = O(n^\alpha)$ for some $\alpha < \frac{k-2}{4(k-1)}$.
Then there is some $\lambda > 0$ so that with high probability, an unsatisfiable formula drawn from $F_k(n,m,c,p)$ does not have a resolution refutation of size $2^{O\left( n^\lambda \right)}$.
\end{restatable}

\vspace{-1.5ex} 
\begin{restatable}{theorem}{thmBoundKThree} \label{thm:k3}
Suppose $m = O(n)$ and $c = O(n^\alpha)$ for some $\alpha < 1/10$.
Then there is some $\lambda > 0$ so that with high probability, an unsatisfiable formula drawn from $F_3(n, m, c, p)$ does not have a resolution refutation of size $2^{O\left(n^\lambda\right)}$.
\end{restatable}

\subsection{Deducing a Lower Bound on CDCL Runtime} \label{sec:cdcl-runtime}

Finally, we can conclude that unsatisfiable random instances from $F_k(n,m,c,p)$ with sufficiently few communities usually take exponential time for CDCL to solve.
\begin{theorem}
If $m = O(n)$ and $c = O(n^\alpha)$ for any $\alpha < 1/10$, the runtime of CDCL on an unsatisfiable formula $\varphi$ from $F_k(n,m,c,p)$ is exponential with high probability.
\end{theorem}
\begin{proof}
If $\varphi$ is unsatisfiable, the runtime of CDCL on $\varphi$ is lower bounded (up to a polynomial factor) by the length of the shortest resolution refutation of $\varphi$ \cite{cdcl-sim}.
If $k = 3$, then the shortest refutation of $\varphi$ is exponentially long with high probability by Theorem \ref{thm:k3}.
If instead $k \ge 4$, the same is true by Theorem \ref{thm:k4}, since $1/10 < \frac{k-2}{4(k-1)}$.
Therefore with high probability, CDCL will take exponential time to prove $\varphi$ unsatisfiable.
\qed
\end{proof}

\begin{remark}
By picking a sufficiently high clause-variable ratio, we can ensure $\varphi$ is unsatisfiable with high probability, so that CDCL takes exponential time on average for formulas drawn from $F_k(n,m,c,p)$ (not just the unsatisfiable ones).
\end{remark}

We also note that our proof technique is not sensitive to the details of how the community attachment model is defined.
For example, changing the definition of a bridge clause so that the variables do not all have to be in different communities requires only minor changes to the proof (detailed in Appendix \ref{sec:more-proofs}).
\begin{restatable}{theorem}{thmModifiedModel}
Let $\widetilde{F}_k(n,m,c,p)$ be the community attachment model modified so that any clause that is not within a single community counts as a bridge clause.
Then if $m = O(n)$ and $c = O(n^\alpha)$ for any $\alpha < 1/10$, the runtime of CDCL on an unsatisfiable formula $\varphi$ from $\widetilde{F}_k(n,m,c,p)$ is exponential with high probability.
\end{restatable}

Thus we have showed that similarly to unsatisfiable random $k$-{\SAT} instances, unsatisfiable random \emph{modular} instances (as formalized by the community attachment model) are hard on average for CDCL as long as they do not have too many communities.

\section{Discussion} \label{sec:discussion}

We have introduced a broad class of ``modularity-like'' graph metrics, the polynomial clique metrics, and showed that no PCM can be a guaranteed indicator of whether a {\SAT} instance is easy (unless $\mathsf{P}=\NP$).
This is perhaps not too surprising in light of the fact that the VIG throws away the Boolean information in the formula.
While the VIG has received the most attention in recent work on community structure, it would be worthwhile to investigate other graph encodings that preserve more information.
Regardless, our result does indicate that it may be difficult to define a tractable class of {\SAT} instances based purely on modularity or its variants.
Furthermore, by setting up a concrete barrier (the PCM property) that must be avoided to obtain a tractable class, our result can help guide future attempts to find a graph metric that does work.

Our result on the community attachment model $F_k(n,m,c,p)$ is more interesting, as it shows that instances from this model are exponentially hard for CDCL even on average (when $c$ is small enough).
An important point is that the result is actually nontrivial when $p < 1$, unlike for $p = 1$.
In the latter case there are no bridge clauses, so the instances consist of $c$ independent problems of size $n/c$, and since we assume $c = O(n^{1/10})$ each problem has size $\Omega(n^{9/10})$.
So by the old results on random $k$-{\SAT} CDCL would take exponential time to solve even the easiest problem, and so likewise for the original instance (with a slightly smaller exponent on $n$).
On the other hand, when $p < 1$ it is conceivable that the bridge clauses could actually make the instances easier, by adding some extra propagation power or easier-to-find contradictions that would make the whole instance easier to solve than any individual community.
Our result effectively says that this happens with vanishing probability as $n \rightarrow \infty$.

The case $p = 1$ also brings out an important caveat when interpreting our result as evidence that community structure doesn't explain CDCL's effectiveness on industrial instances.
Our result shows that such structure isn't enough to bring random formulas down from exponential-time-on-average to polynomial-time-on-average.
However, it could decrease the time from (say) $2^{n^{1/2}}$ to $2^{n^{1/4}}$, which could be the difference between intractability and tractability if $n$ is small enough.
On the other hand, given the enormous size of many industrial instances it isn't clear whether this is really all that is happening.
It would be interesting to do experiments on parametrized families of industrial instances to see whether CDCL actually avoids exponential behavior, or if the point of blow-up is just pushed out far enough that we tend not to encounter it in practice.

Another important aspect of our result is the limit on the number of communities.
It does not apply when communities have logarithmic size, for example, so that $c = \Theta(n / \log n)$.
In fact it is easy to see that the result cannot hold in this case: if one of the communities is unsatisfiable then it will have a polynomial-length resolution refutation, and as $c \rightarrow \infty$ the probability that at least one community is unsatisfiable by itself goes to $1$.
So with high probability the entire instance has a short refutation, and CDCL could in theory solve it in polynomial time.
A clear direction for future work is to see whether improved proof techniques can extend our results to larger numbers of communities, closing the gap between $O(n^{1/10})$ and $\Omega(n / \log n)$.
This is also another way our results can inform future experiments: it would be interesting to explore a variety of growth rates for $c$ above $n^{1/10}$ and see how the performance of CDCL changes.

We have done some preliminary experiments along these lines, sampling instances from $F_3(n,5n,5n^\alpha,0.9)$ for a variety of values of $n$ and solving them with MiniSat 2.2.0 \cite{minisat}.
For each value of $n$ we generated 10 instances using the generator from \cite{giraldez2015modularity} and averaged their runtimes.
Note that every instance had at least 10 communities, so that the expected modularity was at least $p - (1/c) \ge 0.9 - 0.1 = 0.8$.
In Figure \ref{fig:runtime-community-size}, we plot the results as a function of the community size $n / (5n^\alpha) = n^{1-\alpha} / 5$.
\begin{figure}[ht]
\centering
\includegraphics[width=\textwidth]{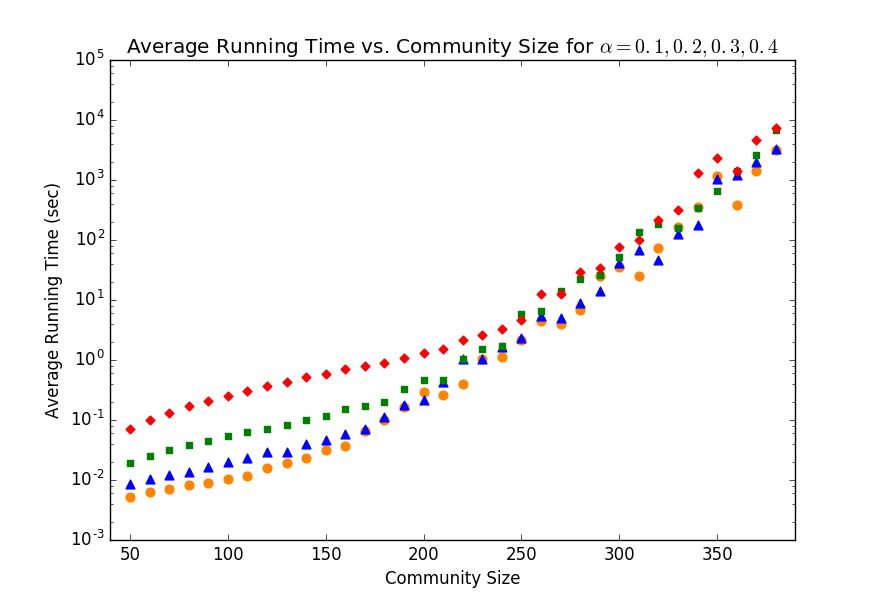}
\caption{Average MiniSat runtimes for instances from $F_3(n,5n,5n^\alpha,0.9)$ for $\alpha = 0.1$ (orange), $0.2$ (blue), $0.3$ (green), and $0.4$ (red), plotted as a function of community size ($n^{1-\alpha}/5$).}
\label{fig:runtime-community-size}
\end{figure}
It is clear from the right half of the graph\footnote{On the left half the growth is much slower and close to linear, but since this occurs only for runtimes on the order of a second or less it may be that parsing the formula and initializing the solver dominate the time needed for the actual search.} that the runtime blows up exponentially in the community size (and thus in $n$) even for $\alpha$ significantly larger than $1/10$.
This suggests that improving our result to larger $\alpha$ is likely possible.
However, it is important to point out that although all values of $\alpha$ are undergoing exponential growth, the lines for the different values of $\alpha$ are quite close together, indicating that community size is a much more important factor in determining runtime than the total formula size.
For example, at a community size of around 270 the $\alpha=0.1$ instances have 3,000 variables total, while the $\alpha=0.4$ instances have 165,000.
So the latter are more than 50 times larger than the former, but their runtimes are only about 3 times longer.
This shows that while larger values of $\alpha$ do not avoid exponential blowup, they can significantly aid performance.

In total, our results indicate that high modularity alone may not be adequate to ensure good performance even on average, but that it could be rewarding to investigate more refined notions of ``good community structure'' that somehow restrict the number of communities.

\noindent
{\bf Acknowledgments.} 
The authors thank Vijay Ganesh, Holger Hoos, Zack Newsham, Markus Rabe, Stefan Szeider, and several anonymous reviewers for helpful discussions and comments.
This work is supported in part by the National Science Foundation Graduate Research Fellowship Program under Grant No. DGE-1106400, by the Hellman Family Faculty Fund, by gifts from Microsoft and Toyota, and by TerraSwarm, one of six centers of STARnet, a Semiconductor Research Corporation program sponsored by MARCO and DARPA.

\bibliographystyle{splncs}
\bibliography{main}

\newpage

\appendix

\section{Other Clustering Metrics} \label{sec:other-metrics}

In this section we give additional evidence that the PCMs form a broad and interesting class of metrics by showing that several popular clustering metrics other than modularity are PCMs.
While these metrics have not previously been used in the context of {\SAT}, they have been widely used in other fields to measure the quality of vertex partitions \cite{best-mod}.
Note that when necessary, we have slightly adjusted the definitions to yield values in $[0,1]$ so that they are graph metrics according to the definition in Section \ref{sec:worst-case}.

First we establish some common notation.
\begin{definition}
Let $P$ be a path in a weighted graph from $s$ to $t$.
Then
\[
d(s, t) = \min_P \sum_{(u, v) \in P} w(u, v) 
\]
is the weight of the shortest path from $s$ to $t$.
If there is no path from $s$ to $t$, let $d(s, t) = \infty$, defining for convenience that $\infty / \infty = 1$.
\end{definition}
\begin{definition}
For any vertex partition $\delta$, the \textbf{community} $C_v$ of $v$ is the equivalence class of $v$ under $\delta$.
\end{definition}
In this section the variable $\delta$ always ranges over vertex partitions, which when convenient we view as a set of communities.

Now we can define the \emph{silhouette index}, which measures how separated the communities are from each other \cite{best-mod}.
\begin{definition}
Let
\[
a(v) = \frac{1}{|C_v|} \sum_{t \in C_v, t \ne v} d(v, t)
\]
be the average distance between $v$ and the other vertices in the same community.
Let
\[
b(v) = \min_{C_v' \not = C_v} \frac{1}{|C_{v'}|} \sum_{t \in C_{v'}} d(v, t)
\]
be the average distance between $v$ and the vertices in the closest other community.
Then the \textbf{silhouette index} of a graph $G = (V, E)$ is
\[
S = \frac{1}{2} \left(1 + \max_\delta \frac{1}{|V|} \sum_v \frac{b(v) - a(v)}{\max(a(v), b(v))} \right).
\]
\end{definition}

Next we define the (external) \emph{conductance}, which compares the weight of the edges spanning communities to the weight of the edges within communities \cite{best-mod}.
\begin{definition}
If for any $S \subseteq V$ we define
\[
r(S) = \sum_{x \in S} \sum_{y \in V} w(x, y) ,
\]
then the \textbf{conductance} of $G$ is
\[
R = 1 - \max_{\delta} \frac{1}{|\delta|} \sum_{C \in \delta} \frac{\sum_{x \in C} \sum_{y \not \in C} w(x, y)}{\min(r(C), r(V \setminus C))}.
\]
\end{definition}

Another simple metric is \emph{coverage}, which compares the weight of the edges within communities to the total weight of the graph \cite{best-mod}.
\begin{definition} The \textbf{coverage} of a graph is
\[
Cov = \max_\delta \frac{\sum_{u \in V} \sum_{v \in C_u} w(u, v) }{\sum_{u, v \in V} w(u, v)}.
\]
\end{definition}

Finally, we define \emph{performance}, which is a sum of two terms: the number of edges within communities, and the number of \emph{missing} edges between communities \cite{best-mod}.
\begin{definition}
The \textbf{performance} of an unweighted graph $G = (V, E)$ is
$$\text{Perf} = \max_\delta \frac{|\{(u, v) \in E :  u \in C_v\}|  + |\{(u, v) \not \in E : u \not \in C_v\}|}{n(n - 1)}.$$
\end{definition}
Note that we only consider the unweighted version of performance, as weighted versions require additional contextual information in the form of reasonable guesses for the weights of missing edges.
For all the other metrics above, an unweighted version can be obtained simply by assuming unit weights.

Now we prove that all of these metrics are PCMs.
In fact this is for a trivial reason: they all consider a single complete graph to be well-clustered.
So an attempt to use them as difficulty metrics in the context of {\SAT} would need to change their definitions, for example by maximizing only over $\delta$ with at least some minimum number of communities.
\otherMetricsPCM*
\begin{proof}[sketch]
Take $c = 1$, so that $G = K_n$.
Consider a partition $\delta$ which puts all vertices into the same community.
Then there are no vertices outside of that community, so the silhouette index is 1.
Similarly, there are no edges between communities, so the conductance and coverage are both 1.
Finally, since there are $n(n-1)/2$ edges inside the community, the performance is 1.
\qed
\end{proof}

\section{Deferred Proofs} \label{sec:more-proofs}

\thmBoundKFour*
\begin{proof}
Fix some $\lambda, \epsilon > 0$ which we will require to be sufficiently small later, and define $x(n) = n^{(\lambda - 1) / 2}$.
To apply Theorem \ref{main} we must have $t(n) = \omega(1)$ and 
\[
t(n) \leq \min\{x(n) n / 2, s(n) / z(n) \}
\]
where $s(n) = x(n)^{1 - 1 / y(k - 1)}n^{1 - 1 / (k - 1)}$.
So we set $t(n) = x(n) n / 2$ and prove that for an appropriate $z(n) = \omega(1)$, this is less than $s(n) / z(n)$ for sufficiently large $n$.
First, note that
\[
x(n)^{ - 1 / y} = n^{\frac{-(\lambda - 1)}{2y}} = n^{\frac{(1 - \lambda)(k + \epsilon)}{4}} = \omega(n)
\]
for sufficiently small $\lambda$, since $k \ge 4$ and $\epsilon > 0$.
Therefore,
\[
s(n) = x(n)^{1 - \frac{1}{y(k - 1)}} n^{1 - \frac{1}{k - 1}} = x(n) \left( x(n)^{-1/y} \right)^{1/(k-1)} n^{1 - \frac{1}{k - 1}} \in \omega \left( x(n) n \right).
\]
Since $t(n) = x(n) n / 2$, there is some $z(n)$ in $\omega(1)$ such that $s(n) / z(n) \geq t(n)$ for sufficiently large $n$.
Finally,
\[
t(n) = n^{(\lambda - 1) / 2}n / 2 = \Theta\left(n^{(\lambda + 1) / 2}\right) \subseteq \omega(1)
\]
since $\lambda > 0$.
Thus we have satisfied all the conditions of Theorem \ref{main}.
Now observe that
\begin{align*}
\frac{1}{2^{7k / 2} (1 + \epsilon)c^{k - 1}}x(n)^{-(k - 2 - \epsilon) / 2}n &= \Omega \left(\frac{1}{n^{\alpha(k - 1)}}n^{-(\lambda - 1)(k - 2 - \epsilon) / 4}n \right) \\
&= \Omega \left(n^{1 + \frac{(1 - \lambda)(k - 2 - \epsilon)}{4} - \alpha(k-1)} \right)\\
&= \omega \left(n^{1 + \frac{(1 - \lambda)(k - 2 - \epsilon)}{4} - \frac{k-2}{4}} \right) \\
&= \omega \left(n^{1 + \frac{1}{4} \left[ (1 - \lambda)(k - 2 - \epsilon) - (k-2) \right]} \right) \\
&= \omega(n)
\end{align*}
for sufficiently small $\lambda$ and $\epsilon$.
So for sufficiently large $n$ this quantity is larger than $m = O(n)$, and Theorem \ref{main} applies to $F_k(n,m,c,p)$.
Therefore with high probability, an unsatisfiable instance from $F_k(n, m, c, p)$ does not have a resolution refutation of size $2^{O(x(n)t(n))} = 2^{O(x(n)^2 n / 2)} = 2^{O\left( n^\lambda \right)}$.
\qed
\end{proof}

\thmBoundKThree*
\begin{proof}
We proceed along the lines of the previous theorem, except that we set $t(n) = s(n) / z(n) = x(n)^{1 - 1 / 2y}n^{1/2}/z(n) = x(n)^{\frac{1 - \epsilon}{4}}n^{1/2} / z(n)$ and show that $t(n) \le x(n) \cdot n / 2$ for sufficiently large $n$.

Fix some $\gamma, \epsilon > 0$ which we will require to be sufficiently small later.
Letting $z(n) = n^\gamma$, clearly $z(n) = \omega(1)$.
Also define
\[
x(n) = \left( n^{-\gamma - \alpha \frac{5-\epsilon}{1-\epsilon}} \right)^{\frac{4}{5-\epsilon}} .
\]
Then
\[
t(n) = \left( n^{-\gamma - \alpha \frac{5-\epsilon}{1-\epsilon}} \right)^{\frac{4}{5-\epsilon} \cdot \frac{1-\epsilon}{4}} n^{1/2} / z(n) = n^{\frac{1}{2} -\gamma \left( 1 + \frac{1-\epsilon}{5-\epsilon} \right) - \alpha} ,
\]
which is $\omega(1)$ for sufficiently small $\gamma$ since $\alpha < 1/10$.
Also note that
\[
x(n) t(n) = x(n)^{\frac{5-\epsilon}{4}} n^{\frac{1}{2} - \gamma} = n^{\frac{1}{2} - 2\gamma - \alpha \frac{5-\epsilon}{1-\epsilon}} = n^\lambda
\]
where $\lambda = \frac{1}{2} - 2\gamma - \alpha \frac{5-\epsilon}{1-\epsilon}$.
Again since $\alpha < 1/10$, we have $\lambda > 0$ for sufficiently small $\gamma$ and $\epsilon$.
Similarly,
\[
x(n)^{-1/y} = \left( n^{-\gamma - \alpha \frac{5-\epsilon}{1-\epsilon}} \right)^{- \frac{4}{5-\epsilon} \cdot \frac{3+\epsilon}{2}} = n^{2\gamma \frac{3+\epsilon}{5-\epsilon} + 2\alpha \frac{3+\epsilon}{1-\epsilon}} = o(n)
\]
for sufficiently small $\gamma$ and $\epsilon$.
Therefore
\[
t(n) = x(n) \left( x(n)^{-1/y} \right)^{1/2} n^{1/2} / z(n) = o(x(n) n) .
\]
So for sufficiently large $n$ we have satisfied the conditions of Theorem \ref{main}.
Now observe that
\begin{align*}
\frac{1}{2^{21 / 2} (1 + \epsilon)c^{2}}x(n)^{-(1 - \epsilon) / 2}n &= \Omega \left(\frac{1}{n^{2\alpha}} \left( n^{-\gamma - \alpha \frac{5-\epsilon}{1-\epsilon}} \right)^{-\frac{4}{5-\epsilon} \cdot \frac{1-\epsilon}{2}} n \right) \\
&= \Omega \left( n^{1 + 2\gamma \frac{1-\epsilon}{5-\epsilon}} \right) \\
&= \omega(n)
\end{align*}
since $\gamma > 0$ and we may take $\epsilon < 1$.
So for sufficiently large $n$ this quantity is larger than $m = O(n)$, and Theorem \ref{main} applies to $F_3(n,m,c,p)$.
Therefore with high probability, an unsatisfiable instance from $F_3(n,m,c,p)$ does not have a resolution refutation of size $2^{O(x(n) t(n))} = 2^{O\left(n^\lambda\right)}$.
\qed
\end{proof}

\thmModifiedModel*
\begin{proof}
Modify Algorithm \ref{mod_levy} to use the new definition of bridge clause, store in $b$ the new bridge clause probability $1 - h$, and sample from $\widetilde{F}_k(n,m,c,p)$ on line \ref{line:sample-psi}.

Now we check that each lemma is still true.
For Lemma \ref{lemma:well-defined}, observe that removing the constraint that every variable in a bridge clause must come from a different community cannot decrease the probability that a random clause is a bridge clause.
So our new value of $b$ is at least as large as the old, and therefore $h/b$ is still at most $1$.
Also $\widetilde{F}_k(n,m,c,p)$ is a distribution over $k$-CNF formulas with $n$ variables and $m$ clauses, so Lemma \ref{lemma:well-defined} holds.

For Lemma \ref{lemma:equiv}, when the modified Algorithm \ref{mod_levy} returns from line \ref{line:failure} the formula $\psi$ is drawn from $\widetilde{F}_k(n,m,c,p)$.
So in this case $\overline{F}_k(n,m,c,p;m')$ is trivially identical to $\widetilde{F}_k(n,m,c,p)$, and we need only consider the case when the algorithm returns from line \ref{line:success}.
As above each clause of $\phi$ is added independently, so we need only calculate the probability that an added clause is within a community.
Proceeding exactly as in Lemma \ref{lemma:equiv}, we obtain $\Pr[\cadded | A] = \Pr[\ccomm] = h$ and $\Pr[\cadded | \overline{A}] = (h/b) \Pr[\cbridge] = h$ (since we changed $b$ to be the probability of getting a bridge clause under the new definition).
So
\[
\Pr[\cadded] = \Pr[\cadded | A] \Pr[A] + \Pr[\cadded | \overline{A}] \Pr[\overline{A}] = hp + h(1-p) = h,
\]
and therefore
\[
\Pr[\ccomm | \cadded] = \Pr[ A | \cadded] = \frac{\Pr[\cadded | A] \Pr[A]}{\Pr[\cadded]} = \frac{hp}{h} = p
\]
as before.

For Lemma \ref{lemma:success}, note that the probability that $C$ is a bridge clause is the new value of $b$.
So as before the probability that $C$ is added to $\psi$ is $p \cdot h + (1-p) \cdot (h/b) \cdot b = h$.
The rest of the computation then proceeds without change.

Finally, the argument for Lemma \ref{lemma:inherit} goes through with no changes.
So subformula-inherited properties of random $k$-{\SAT} instances are passed on to instances of $\widetilde{F}_k(n,m,c,p)$ with high probability.
Therefore the Beame--Pitassi argument in Section \ref{sec:beame-pitassi} and the CDCL runtime bound in Section \ref{sec:cdcl-runtime} hold for $\widetilde{F}_k(n,m,c,p)$ with no further modifications needed.
\qed
\end{proof}

\end{document}